\documentclass[12pt]{revtex4-1}
\usepackage{amsfonts,amsmath,amssymb,xcolor}
\usepackage{amsthm,amscd,ulem,bm}
\usepackage{graphicx}
\makeatletter

\@addtoreset{equation}{section}
\makeatother
\theoremstyle{definition}
\newtheorem{pr}{Proposition}[section]
\newtheorem{defn}{Definition}[section]
\newtheorem{rem}{Remark}

\newcommand{\bib}[2]{\frac{\partial {#1}}{\partial {#2}}}
\newcommand{\bbib}[3]{\frac{\partial^2 {#1}}{\partial {#2}{\partial {#3}}}}

\begin{document}
\title{
Finsler connection in moving frame formalism
}
\author{Takayoshi Ootsuka}
\email{ootsuka@cosmos.phys.ocha.ac.jp}
\affiliation{Department of Physics, 
Ochanomizu University, 
2-1-1 Ootsuka Bunkyo-ku, Tokyo, Japan }
\author{Muneyuki Ishida}
\email{ishida@phys.meisei-u.ac.jp}
\affiliation{Department of Physics, Meisei University,
2-1-1 Hodokubo, Hino, Tokyo, Japan }
\author{Ryoko Yahagi}
\email{yahagi@rs.tus.ac.jp}
\affiliation{Department of Physics,
Faculty of Science, 
Tokyo University of Science,
1-3 Kagurazaka, Shinjuku-ku, Tokyo, Japan }

\date{\today}

\begin{abstract}

In our previous work, we have defined a nonlinear connection of Finsler manifold
which preserves the Finsler metric $L=L(x,dx)$.
To make 
the method easier and more useful in applications,
moving frame (vielbein) $\theta^a={e^a}_\mu dx^\mu$
formalism for the nonlinear connection is newly considered.
We derive 
formulae 
to calculate the Finsler connection in the
specific case that the Finsler metric depends not on coordinates 
$x^\mu$, but only on 
moving frame $\theta^a$: $L=L(\theta)$.

\end{abstract}

\maketitle

\section{Introduction}

Lagrangian systems are naturally described with Finsler geometry~\cite{Lanczos, Suzuki}
in the same sense that Hamiltonian systems are described with 
symplectic geometry or contact geometry
we can construct an $(n+1)$-dimensional Finsler manifold $(M,L)$
from a Lagrangian system $(Q,L^\ast)$, a pair of 
an $n$-dimensional configuration space $Q$
and a Lagrangian of the system 
$L^\ast=L^\ast\bigl(q^i,\dot{q}^i,t\bigr), \, (i=1,2,\dots,n)$.
\begin{align}
 M=\mathbb{R}\times Q, \quad
 L=L(x,dx):=L^\ast\left(x^i, \frac{dx^i}{dx^0},x^0\right)|dx^0|.
\end{align}
$M$ is called {\it the extended configuration space} in analytical dynamical terminology.
$L$ becomes a Finsler metric,
which is a function of coordinates $(x^\mu)$ and their derivative
(adopted fiber coordinates of $TM$) $(dx^\mu)$
and is homogeneity one with respect to $(dx^\mu)$,
\begin{align}
 L(x,\lambda dx)=\lambda L(x,dx), \quad \lambda >0. \label{homogeneity}
\end{align} 

We do not assume positivity: $L>0$ and regularity
: ${\rm det}\bigl(g_{\mu\nu}(x,dx)\bigr)\neq 0$,
$\displaystyle g_{\mu\nu}(x,dx)=\frac12\bbib{L^2}{dx^\mu}{dx^\nu}$
as in standard textbooks
of Finsler geometry~\cite{Matsumoto, Miron, Chern, Antonelli}.
This is because these restrictions can be burdens of physical applications.
We only assume the positive homogeneity condition
(\ref{homogeneity}).

In our previous work~\cite{KO}, 
we give a definition of a generalized Berwald's nonlinear connection
on Finsler manifold $(M,L)$ in point Finsler perspective~\cite{Kozma}.
In this standpoint, Finsler geometry is considered to be a geometry
not on the line element space (slit tangent bundle $TM^{\circ}:=TM-\{0\}$), but 
on the point space $M$,
in contrast to the conventional perspective where
Finsler connection is defined as a linear connection on $TM^{\circ}$.
For physical applications,
we believe only nonlinear parts of the Finsler connection is needed;
linear parts do not play any role in physics.
Our nonlinear Finsler connection which defines parallel transports
on a point space $M$ can be a natural extension of 
the Levi-Civita connection on the point space $M$
as an event manifold
which Einstein used as a tool of his general relativity.
If you are persistent in constructing a deformed gravity with such a 
conventional Finsler connection,
you need to clarify its physical meaning to the coordinates $(y^\mu)$
of $(x,y) \in TM^\circ$ on which you consider a linear connection.
On the other hand, our nonlinear Finsler connection has been used 
in constrained dynamics~\cite{KO}, 
fluid mechanics~\cite{OTIY}, Killing vectors on Finsler manifold~\cite{OITY}
and superparticles~\cite{OITYsuper}, 
proving itself worthwhile in physical applications. 
Explicit calculation of this nonlinear connection, however, 
tend to be complicated. 
Even in the case of a Finsler metric that is a natural extension of 
Riemannian metric 
(e.g. the super Finsler metric defined from 
a superparticle Lagrangian~\cite{OITY}), 
it is hard to say easy. 
The purpose of this paper is to make the calculation easier 
by using vielbeins, 
following general relativity where the vielbein technique plays a critical role 
in simplification of the theoretical structure as well as calculations. 
In contrast to the vielbein $\theta^a$ in a Riemannian manifold, 
which is chosen to span an orthonormal basis, 
we do not impose orthonormality, since there is no concept of orthogonality in Finsler manifold in general.
A vielbein $\theta^a$ in this paper is just an independent 1-form basis, 
and the Finsler metrics considered here depend on only such a vielbein 
and have no coordinate function $(x^\mu)$ dependence:
$L=L(\theta^a)$.
Similar to the Einstein-Cartan theory of gravity, it is sufficient to take the Finsler metrics
of the above type when Finsler gravity theory is considered.

In the next section, 
we give a short review of our nonlinear Finsler connection.
In sections 3, 4 and 5, we show the description of the nonlinear connection 
in moving frame formalism.
In section 6, the actual calculations of the connections in two examples,
Riemannian manifold and pseudo particle model, are given.

\section{nonlinear Finsler connection}

Let $(M,L)$ be an $(n+1)$-dimensional Finsler manifold. 
Usually a nonlinear connection of $TM$ is defined as a subbundle  
$HTM$ of $TTM$ such as 
\begin{align}
 TTM=HTM \oplus VTM,
\end{align}
where $VTM$ is the vertical vector bundle over $TM$
\cite{Miron, Antonelli}.
With the adopted coordinates $(x^\mu,dx^\mu)$ of $TM$, 
the horizontal bundle $HTM$ is generated by vector fields 
$\displaystyle \frac{\delta}{\delta x^\rho}$;
\begin{align}
 HTM=\left\langle \frac{\delta}{\delta x^\rho}
 :=\bib{}{x^\rho}-{N^\mu}_\rho(x,dx) \bib{}{dx^\mu}
 \right\rangle,
\end{align}
where ${N^\mu}_\rho(x,dx)$ are functions of $x^\mu$ and $dx^\mu$,
and homogeneity one with respect to $dx^\mu$:
${N^\mu}_\rho(x,\lambda dx)=\lambda {N^\mu}_\rho(x,dx), \, \lambda >0$.

Using these coefficitents ${N^\mu}_\rho(x,dx)$,
we can define a nonlinear covariant derivative for vector fields on $M$.
For $X, Y \in \Gamma(TM)$, the covariant derivative $\nabla_X Y$
is defined by
\begin{align}
 \nabla_X Y:= X^\rho \left\{\bib{Y^\mu}{x^\rho} +{N^\mu}_\rho(x,dx(Y))
 \right\} \otimes \bib{}{x^\mu} \in \Gamma(TM).\label{nablaXY}
\end{align}
This nonlinear connection ${N^\mu}_\rho(x,dx)$ actually defines
parallel transports on $M$.
Other linear parts of a linear connection of $TM$ called $N$-linear connection
do not have any use in physics.
Therefore it is convenient to define the nonlinear connection,
with tangent-bundle-terminology excluded.
Furthermore, we can choose a nonlinear connection which preserves Finsler metric:
$\nabla L=0$, and we call it 
nonlinear Finsler connection. 
We first define a nonlinear connection on $\Gamma(T^\ast M)$,
and from duality, we define 
the
covariant derivative on $\Gamma(TM)$ 
as displayed in (\ref{nablaXY}).   

\begin{defn}
We call $\nabla$ a nonlinear Finsler connection on $M$
which satisfies following conditions.
\begin{align}
 &\nabla dx^\mu=- dx^\rho \otimes {N^\mu}_\rho(x,dx), \\
 &\bib{{N^\mu}_\rho}{dx^\nu}-\bib{{N^\mu}_\nu}{dx^\rho}=0, \label{Cartan} \\
 &\bib{L}{x^\rho}-\bib{L}{dx^\mu}{N^\mu}_\rho=0.\label{metricity}
\end{align}
\end{defn}

If the Hessian of the Finsler metric $L$ with respect to $dx$,
$\displaystyle L_{\mu\nu}:=\bbib{L}{dx^\mu}{dx^\nu}$,
satisfies ${\rm rank}(L_{\mu\nu})=n$, 
the Finsler metric is called {\it regular}, and if
${\rm rank}(L_{\mu\nu}) < n$, it is called {\it singular}.
Our nonlinear connection can be defined in the case of 
singular Finsler metrics
as well as regular ones.

We can prove the following existence theorem of such a nonlinear connection~\cite{KO}.
\begin{pr}
 If
 ${\rm rank}\bigl(L_{\mu\nu}\bigr)=n-D$, take coordinates as
 $(x^\mu)=(x^0,x^a,x^I), \, (a=1,2,\dots,n-D), \, (I=n-D+1,n-D+2, \dots ,n)$ such that
 ${\rm det}\bigl(L_{ab}\bigr)\neq 0$.
Then, nonlinear connection coefficients ${N^\mu}_\rho$
 are given by the derivatives of the functions $G^\mu, \, (\mu=0,a,I)$ defined below 
with respect to $dx^\rho$:
$\displaystyle{N^\mu}_\rho=\bib{G^\mu}{dx^\rho}$,
\begin{align}
 G^\mu:=\frac12 \left(dx^\rho \bib{L}{x^\rho}\right)\frac{dx^\mu}{L}
 +L^{ab}M_b L\bib{}{dx^a}\left(\frac{dx^\mu}{L}\right)
 +\lambda^I v^\mu_I, 
 \quad 
 M_I-L^{ab}L_{aI}M_b=0,
\end{align}
where the matrices $\bigl(L^{ab}\bigr)$ are 
the inverse of the matrices $(L_{ab})$, and $M_\mu, \, (\mu=a,I)$ are defined by
\begin{align}
  M_\mu:=\frac12 \left(-\bib{L}{x^\mu}+dx^\rho \bbib{L}{dx^\mu}{x^\rho}\right),  
\end{align}
and $\lambda^I$ are arbtrary homogeneity-two functions
and 
$v^\mu_I$ are such 0-eigenvalued homogeneity-zero functions 
that 
$L_{\mu\nu} \, v^\nu_I=0$.
\end{pr}
\begin{rem}
 There are $D$ constraints,
 ${\cal C}_I:= M_I-L^{ab}L_{aI}M_b=0$,
 which determine some of the {\it Lagrange multipliers} $\lambda^I$. 
 If all of $\lambda^I$ are determined, we call the systems
 {\it second class constrained systems}.
 Otherwise, ${\cal C}_I$
 include some {\it first class constraints}.
\end{rem}
\begin{rem}
 The regularity condition of Finsler metric $L$, 
 ${\rm rank}\bigl(L_{\mu\nu}\bigr)=n$,
 corresponds to the famous condition of nondegeneracy of the 
 fundamental tensor: 
 $\displaystyle {\rm det}\bigl(g_{\mu\mu}\bigr)\neq 0, \, g_{\mu\nu}(x,dx)
 :=\frac12 \bbib{L^2}{dx^\mu}{dx^\nu}$.
 Then our nonlinear connection ${N^\mu}_\rho(x,dx)$
 becomes the nonlinear parts of the Berwald's 
 connection. 
\end{rem}

The covariant derivative of a vector field $Y \in \Gamma(TM)$
is defined from the derivative of $Y^\mu=\langle dx^\mu, Y\rangle$,
\begin{align}
 dY^\mu=\langle \nabla dx^\mu, Y\rangle+\langle dx^\mu, \nabla Y\rangle
 =-{N^\mu}(x,dx)(Y)+\langle dx^\mu, \nabla Y\rangle,
\end{align} 
where the contraction of ${N^\mu}_\rho(x,dx)$ and $Y$
is defined as ${N^\mu}_\rho\bigl(x,dx(Y)\bigr)$.
Therefore
\begin{align}
 \nabla Y=dY^\mu +{N^\mu}\bigl(x,dx(Y)\bigr), \quad 
 {N^\mu}\bigl(x,dx(Y)\bigr):=dx^\rho \otimes {N^\mu}_\rho\bigl(x,dx(Y)\bigr),
\end{align}
and $\nabla_X Y$ becomes (\ref{nablaXY}).

We will call the condition (\ref{Cartan}) {\it Cartan condition},
which coresponds to the torsionless condition in the case of 
${N^\mu}_\rho(x,dx)$ are linear
with respect to $(dx)$.
But in general, 
this Cartan condition does not mean torsionless because of the nonlinearity of the connection; 
\begin{align}
 \nabla_{\bib{}{x^\rho}} \bib{}{x^\nu}-\nabla_{\bib{}{x^\nu}} \bib{}{x^\rho}
 =\bib{{N^\mu}_\rho}{dx^\nu}
 \left(x,dx\left(\bib{}{x^\rho}\right)\right)
 -\bib{{N^\mu}_\nu}{dx^\rho}\left(x,dx\left(\bib{}{x^\rho}\right)\right) \neq 0.
\end{align}
{\it Metricity condition}
(\ref{metricity}) shows the preservation of $L$ by $\nabla$;
\begin{align}
 \mbox{``$\displaystyle
 \nabla L=\nabla x^\rho \bib{L}{x^\rho}+\nabla dx^\mu \bib{L}{dx^\mu}$''}
 =dx^\rho \otimes \left(\bib{L}{x^\rho}-{N^\mu}_\rho \bib{L}{dx^\mu} \right)=0.
\end{align}
Here we think that $\nabla$ is a covariant derivative on $M$ not on $TM$,
and we also think $dx^\mu$ as 1-forms on $M$ not as fiber coordinates of $TM$. 
$L$ is not a 1-form but what we call a {\it nonlinear 1-from }.
Therefore $\nabla L$ is calculated through an extention in action 
of $\nabla$ to nonlinear 1-froms.

Furthermore using this nonlinear connection (Berwald function),
we can write the Euler-Lagrange equation in an auto-parallel form~\cite{KO}:
\begin{align}
 0={\cal EL}_\mu(L):=\bib{L}{x^\mu}-d\left(\bib{L}{dx^\mu}\right)
 ~\Leftrightarrow~
 \left\{
 \begin{array}{l}
  d^2 x^\mu+2G^\mu(x,dx)=\lambda dx^\mu, \\
  {\cal C}_I=0
 \end{array}
 \right.
\end{align}

As we mentioned in the introduction, we use only nonlinear parts of the
so-called Finsler connection.
This makes the calculations of Finsler connection, 
which is normally thought to be very difficult, easier.
In fact, the explicit derivation of nonlinear connection with the formula above
is much faster than the one with the conventional method.
What we explained so far is based on the holonomic coordinates.
We further assure that usage of moving frame 
(vielbein)
reduces several more calculation steps.
In the next and the following section, we rewrite the Cartan condition (\ref{Cartan}),
the metricity condition
(\ref{metricity})
, and the Berwald functions
in terms of moving frame.

\section{Cartan and Metricity conditions}

On Riemannian manifold $(M,g)$, we frequently assume 
the torsionless condition: $T=0$, and
the metricity condition: $\nabla g=0$.
As is well known, using moving 
frame (vielbein) $\bigl(\theta^a\bigr)$,
the torsionless condition, 
\begin{align}
 0=T\left(\bib{}{x^\rho},\bib{}{x^\nu}\right)
 =\nabla_{\bib{}{x^\rho}} \bib{}{x^\nu} 
 -\nabla_{\bib{}{x^\nu}} \bib{}{x^\rho}
 =\left(
 {\varGamma^\mu}_{\nu\rho}-{\varGamma^\mu}_{\rho\nu}
 \right)\bib{}{x^\mu},
\end{align}
is equivalent to the condition,
\begin{align}
 d\theta^a + {\omega^a}_b \wedge \theta^b=0,
\end{align}
where ${\omega^a}_b$ are spin connections defined by
$\nabla \theta^a=-{\omega^a}_b \theta^b$.
If the Riemannian metric is written as
$g=\mu_{ab}\theta^a \otimes \theta^b$ with constant coefficients $\mu_{ab}$,
the metricity condition $\nabla g=0$ corresponds to the 
following condition,
\begin{align}
 \omega_{ab}+\omega_{ba}=0, \quad \omega_{ab}:=\mu_{ac}{\omega^c}_b.
\end{align}

Similarly,
we derive the corresponding conditions 
for the nonlinear connections on a Finsler manifold.
First, we should replace the Cartan's condition (\ref{Cartan}),
which corresponds to the torsionless condition on a Riemannian manifold,
in terms of moving frames.
${N^\mu}_\rho$ are written as
\begin{align}
 {N^\mu}_\rho=-\nabla_{\rho}dx^\mu=-\nabla_{\rho}\bigl(\theta^a {e_a}^\mu\bigr)
 =-\theta^a \bib{{e_a}^\mu}{x^\rho}+{N^a}_\rho {e_a}^\mu
 =-\theta^a \bib{{e_a}^\mu}{x^\rho}+{N^a}_c {e_a}^\mu {e^c}_\rho,
\end{align}
where ${N^a}_\rho:=-\nabla_\rho \theta^a$ and ${N^a}_c:={N^a}_\rho {e_c}^\rho$,
and which give
\begin{align}
 \bib{{N^\mu}_\rho}{dx^\nu}
 =\bib{\theta^b}{dx^\nu}\bib{{N^\mu}_\rho}{\theta^b}
 ={e^b}_\nu \bib{{N^\mu}_\rho}{\theta^b}
 =-{e^a}_\nu \bib{{e_a}^\mu}{x^\rho}
 +{e^b}_\nu \bib{{N^a}_c}{\theta^b}{e_a}^\mu {e^c}_\rho.
\end{align}
Therefore the Cartan condition (\ref{Cartan}) becomes
\begin{align}
 {e^a}_\rho \bib{{e_a}^\mu}{x^\nu}-{e^a}_\nu \bib{{e_a}^\mu}{x^\rho}
 +{e^b}_\nu \bib{{N^a}_c}{\theta^b}{e_a}^\mu {e^c}_\rho
 -{e^b}_\rho \bib{{N^a}_c}{\theta^b}{e_a}^\mu {e^c}_\nu=0.
\end{align}
Multiplying this by $e^d{}_\mu$,
\begin{align}
 -\bib{{e^d}_\rho}{x^\nu}+\bib{{e^d}_\nu}{x^\rho}
 +\bib{{N^d}_c}{\theta^b}{e^b}_\nu {e^c}_\rho
 -\bib{{N^d}_c}{\theta^b}{e^b}_\rho {e^c}_\nu=0,
\end{align}
therefore we get
\begin{align}
 \bib{{N^a}_c}{\theta^b}-\bib{{N^a}_b}{\theta^c}
 ={e_b}^\nu {e_c}^\rho \left(\bib{{e^a}_\rho}{x^\nu}
 -\bib{{e^a}_\nu}{x^\rho}\right)
 =-{e^a}_\rho e_b\bigl({e_c}^\rho\bigr)+{e^a}_\rho e_c\bigl({e_b}^\rho\bigr).
 \label{Cartan1}
\end{align}
Because the RHS of (\ref{Cartan1}) can be written as
$-\theta^a\bigl([e_b,e_c]\bigr)$,
the Cartan condition also becomes
\begin{align}
 \bib{{N^a}_c}{\theta^b}-\bib{{N^a}_b}{\theta^c}
 +\theta^a\bigl([e_b,e_c]\bigr)=0. \label{Cartan2}
\end{align}

\begin{rem}
 If we define a ``nonlinear spin connection 1-form''
 $\displaystyle{\tilde{\omega}^a}{}_b:=\theta^c \otimes \bib{{N^a}_c}{\theta^b}$,
 then the above Cartan condition is written like the case of Riemannian,
\begin{align}
 {\rm rot \,} \theta^a+{\tilde{\omega}^a}{}_b \wedge \theta^b
 :={\rm rot \,} \theta^a \otimes 1
 +\theta^c \wedge \theta^b \otimes \bib{{N^a}_c}{\theta^b}=0,
\end{align}
where ${\rm rot}$ means the exterior differentiation of 1-forms which becomes 2-forms.
\end{rem}

Second, we will think the metricity condition: $\nabla L=0$.
If we assume the Finsler metric $L$ is written 
only by moving frame 
$\bigl( \theta^a\bigr)$
as $L=L(\theta)$, 
then the metricity condition is rewritten by
\begin{align}
0 &= 
 \nabla L 
 = \nabla \theta^a \bib{L}{\theta^a}
 = -\theta^c \otimes {N^a}_c \bib{L}{\theta^a}, 
 \quad
 \mbox{i.e.} \quad {N^a}_c L_a=0, \quad
 L_a:= \bib{L}{\theta^a}.
 \label{metricity2}
\end{align} 
In the next section, 
using our conditions (\ref{Cartan2}) and (\ref{metricity2}),
we construct a formulae of the Berwald's functions
in moving frame formulation.

\section{Explicit Formula}

From a given Finsler metric $L=L(\theta)$,
we construct the coefficients ${N^a}_c$ of the nonlinear connection,  
which are the functions of $(x^\mu, \theta^a)$,
homogeneity one with respect to $(\theta^a)$,
and satisfy (\ref{Cartan2}) and (\ref{metricity2}).
We define auxiliary functions 
\begin{align}
 G^a := \frac{1}{2} \, \theta^c {N^a}_c,
\end{align}
which we also call Berwald functions.
Note that the product $\theta^c$ and ${N^a}_c$ is a simple multiplication,
not a tensor product.
These Berwald functions satisfy
\begin{align}
 \bib{G^a}{\theta^c}
 &=\frac12 {N^a}_c +\frac12 \bib{{N^a}_b}{\theta^c}\theta^b
 =\frac12 {N^a}_c+\frac12 \left\{
 \bib{{N^a}_c}{\theta^b} \theta^b+\theta^a\bigl([e_b,e_c]\bigr) \theta^b
 \right\} \nonumber \\
 &={N^a}_c + \frac12 \theta^a\bigl([e_b,e_c]\bigr) \theta^b,
\end{align}
where the Cartan condition (\ref{Cartan2}) 
and the first homogeneity of ${N^a}_c$ with respect to $\theta$
have been exploited.
Conversely, if ${N^a}_c$ are given by
\begin{align}
 {N^a}_c=\bib{G^a}{\theta^c}-\frac12 \theta^a\bigl([e_b,e_c]\bigr)\theta^b,\label{GtoN}
\end{align}
we can easily check that these ${N^a}_c$ satisfy the Cartan condition (\ref{Cartan2}). 

Using the Berwald functions, we solve the metricity condition (\ref{metricity2})
with the following assumption of the Finsler metric:
\begin{align}
 L=L(\theta), \quad
 {\rm rank}\bigl(L_{ab}\bigr)=n-D, \quad 
 {\rm det}\bigl(L_{\bm{a}\bm{b}}\bigr)\neq 0, \quad
 \bigl(L_{ab}\bigr):=
 \left(\bbib{L}{\theta^a}{\theta^b}\right),
\end{align}
where 
$a,b=0,1,2,\dots, n$, 
$\bm{a},\bm{b}=1,2,\dots,n-D$ and 
$D$ is a constant for $0 \leq D \leq n$.
Under these assumptions of $L$,
there are $n+1$ independent vectors 
$\displaystyle \left(\frac{\theta^a}{L},{\ell_{\bm{b}}}^a, {v_I}^a \right)$
such that
\begin{align}
 {\ell_{\bm{b}}}^a:=L\bib{}{\theta^{\bm{b}}} \left( \frac{\theta^a}{L} \right)
 =\delta^a_{\bm{b}}-\frac{\theta^a}{L} L_{\bm{b}},
\end{align}
and $D$ vectors ${v_I}^a$, $I=n-D+1,n-D+2,\dots,n$ are defined as eigenvectors
for zero eigenvalue of $\bigl(L_{ab}\bigr)$
and are orthogonal to $L_a$: 
\begin{align}
 L_{ab}\, {v_I}^b=0, \quad L_a\, {v_I}^a=0.
 \label{l_I^a}
\end{align}
\begin{rem}\label{rem-l}
 These basis vectors have the following properties,
\begin{align}
 L_a \, \frac{\theta^a}{L}=1, \quad 
 L_a \, {\ell_{\bm{b}}}^a=0, \quad
 L_a \, {v_I}^a=0, \quad
 L_a \, \bib{{v_I}^a}{\theta^b}=0.
\end{align}
The first three equations are obvious.
The fourth is derived from (\ref{l_I^a}) and the differentiation of the third equation
with respect to $\theta^b$.
\end{rem}

\begin{pr}
$\displaystyle \left(
\frac{\theta^a}{L},{\ell_{\bm{b}}}^a, {v_I}^a
\right)
$ 
are functionally independent. 
\end{pr}
\begin{proof}
Suppose the following linear combination of these vectors is zero:
\begin{align}
 A \, \frac{\theta^a}{L}
 +B^{\bm{b}} \, {\ell_{\bm{b}}}^a
 +C^{I} \, {v_I}^a
 =0. \label{ABC}
\end{align}
Contracting $L_a$ to the above equation, we have $A=0$, 
since $L_a \, {\ell_{\bm{b}}}^a=L_a \, {v_I}^a=0$.
Substituting $A=0$ into (\ref{ABC}), differentiating it with respect to $\theta^b$,
and contracting $L_a$ to it, we obtain
\begin{align}
 B^{\bm{b}}L_a \bib{{\ell_{\bm{b}}}^a}{\theta^b}+C^I L_a \bib{{v_I}^a}{\theta^b}=0,
\end{align}
where the second term of the lef hand side is zero because of (\ref{l_I^a}).
The rest becomes
\begin{align}
 B^{\bm{b}} L_a
 \bib{}{\theta^b}\left(\delta_{\bm{b}}^a-\frac{\theta^a L_{\bm{b}}}{L}\right)
 =B^{\bm{b}}L_a \left(-\frac{\delta^a_b L_{\bm{b}}}{L}
 +\frac{\theta^a L_{\bm{b}}L_b}{L^2}-\frac{\theta^a L_{\bm{b}b}}{L}\right)
 =-B^{\bm{b}}L_{\bm{b}b}=0.
\end{align}
Due to the assumption ${\rm det}\bigl(L_{\bm{a}\bm{b}}\bigr)\neq 0$,
the above equation reduces to $B^{\bm{b}}=0$.
Then we have $C^I {v_I}^a=0$ for independent ${v_I}^a$, 
which leads to $C^I=0$.
\end{proof}

\begin{pr} \label{formula}
 The Berwald functions $G^a$ and the constraints ${\cal C}_I$
 for the Finsler metric $L=L(\theta)$
 whose Hessian $L_{ab}$ has constant rank, ${\rm rank}\bigl(L_{ab}\bigr)=n-D$,
 are given by
\begin{align}
 &G^a=
 -\frac12 \left(\delta^a_{\bm{b}}-\frac{L_{\bm{b}}\theta^a}{L}\right)
 L^{\bm{b}\bm{c}}L_f\, \theta^f\bigl([e_g,e_{\bm{c}}]\bigr)\theta^g
 +\lambda^I {v_I}^a, \\
 &{\cal C}_I=L_f\theta^f\bigl([e_g,e_I]\bigr)\theta^g
 -L_{I\bm{b}}L^{\bm{b}\bm{c}}L_f\, \theta^f\bigl([e_g,e_{\bm{c}}]\bigr)\theta^g=0.
\end{align}
\end{pr}

\begin{proof}
Contracting $\theta^c$ to (\ref{metricity2}), we get an equation for the
Berwald functions:
\begin{align}
 L_a G^a=0.\label{metricity33}
\end{align}
On the other hand, 
by using these basis vectors $\displaystyle \left(
\frac{\theta^a}{L}
,{\ell_{\bm{b}}}^a,{v_I}^a
\right)
$, 
$G^a$ can be expanded as
\begin{align}
 G^a = \lambda^0\ \frac{\theta^a}{L}
 +\lambda^{\bm b}{\ell_{\bm{b}}}^a+\lambda^I {v_I}^a.
 \label{Berwald0}
\end{align}
The requirement (\ref{metricity33}) becomes
\begin{align}
 L_a G^a
 =\lambda^0=0, \label{lambda^0}
\end{align}
due to the properties displayed in Remark \ref{rem-l}.
From (\ref{GtoN}) and (\ref{metricity2}),
\begin{align}
 0 &= L_a {N^a}_c
    = L_a \bib{G^a}{\theta^c}
      -\frac12 L_a \, \theta^a\bigl([e_b,e_c]\bigr) \theta^b \nonumber \\
   &=L_a {\ell_{\bm{b}}}^a\bib{\lambda^{\bm{b}}}{\theta^c} 
    +L_a {v_I}^a \bib{\lambda^I}{\theta^c}
    +\lambda^{\bm{b}} L_a\bib{{\ell_{\bm{b}}}^a}{\theta^c} 
    +\lambda^I L_a \bib{{v_I}^a}{\theta^c} 
    -\frac12 L_a \theta^a\bigl([e_b,e_c]\bigr) \theta^a. 
\end{align}
The first three terms vanish because of Remark \ref{rem-l}.
Then, we get
\begin{align}
 -\frac12 L_a \, \theta^a\bigl([e_b,e_c]\bigr)\theta^b
 =-\lambda^{\bm{b}} L_a 
   \bib{ {\ell_{\bm{b}}}^a}{\theta^c}
  = L_{c{\bm{b}}}\lambda^{\bm{b}}.
\end{align}
Thus we get the following equation,
\begin{align}
 L_{0\bm{b}}\lambda^{\bm{b}}=-\frac12 L_a \, \theta^a\bigl([e_b,e_0]\bigr)\theta^b,
 \quad
 L_{A\bm{b}}\lambda^{\bm{b}}=-\frac12 L_a \, \theta^a\bigl([e_b,e_{A}]\bigr)\theta^b,
 \quad
 L_{\bm{c}\bm{b}}\lambda^{\bm{b}}=-\frac12 L_a \, \theta^a\bigl([e_b,e_{\bm{c}}]\bigr)\theta^b.
 \label{Llambda}
\end{align}
For there is an inverse matrix $\bigl(L^{\bm{a}\bm{b}}\bigr)$
of $\bigl(L_{\bm{a}\bm{b}}\bigr)$;
$L^{\bm{a}\bm{c}} L_{\bm{c}\bm{b}}=\delta^{\bm{a}}_{\bm{b}}$,
we obtain a solution of the third equation of (\ref{Llambda})
\begin{align}
 \lambda^{\bm{b}} &=-\frac12 L ^{\bm{b}\bm{c}} L_a 
 \theta^a \bigl([e_b,e_{\bm{c}}]\bigr) \theta^b,
\end{align}
and constraints from the second equation of (\ref{Llambda}):
\begin{align}
 {\cal C}_I:=\frac12 L_a \, \theta^a\bigl([e_b,e_I]\bigr)\theta^b
 -\frac12 L_{I\bm{b}}L^{\bm{b}\bm{c}}L_a \, \theta^a\bigl([e_b,e_{\bm{c}}]\bigr)\theta^b=0.
 \label{const}
\end{align}
This $\lambda^{\bm{b}}$ with the constraints (\ref{const})
also satisfies the first equation of (\ref{Llambda}),
\begin{align}
 \theta^0 L_{0\bm{b}}\lambda^{\bm{b}}
 &=-\bigl(\theta^I L_{I\bm{b}}+\theta^{\bm{c}}L_{\bm{c}\bm{b}} \bigr)\lambda^{\bm{b}}
 =\frac12 \theta^I L_a \, \theta^a \bigl([e_b,e_I]\bigr)\theta^b
 +\frac12 \theta^{\bm{c}} L_a \, \theta^a \bigl([e_b,e_{\bm{c}}]\bigr)\theta^b
 \nonumber \\
 &=\frac12 \theta^c L_a \, \theta^a \bigl([e_b,e_c]\bigr)\theta^b
 -\frac12 \theta^0 L_a \, \theta^a \bigl([e_b,e_0]\bigr)\theta^b
 =-\frac12 \theta^0 L_a \, \theta^a \bigl([e_b,e_0]\bigr)\theta^b.
\end{align}
Therefore we get the Berwald functions,
\begin{align}
 G^a &
 = - \frac12\left( \delta_{\bm{b}}^a-\frac{L_{\bm{b}}\theta^a}{L}\right) 
   L^{\bm{b}\bm{c}} L_c \, \theta^c \bigl([e_b,e_{\bm{c}}] \bigr) \theta^b
  +\lambda^I {v_I}^a,
 \label{Berwald-formula}
\end{align}
with constraints (\ref{const}).
\end{proof}

Let us consider a simpler case where
the conjugate momenta of $\theta^I$ become $L_I=0$, and
the Hessian of the Finsler metric $L$ becomes
\begin{align}
 \bigl(L_{ab}\bigr)=\left(
 \begin{array}{ccc}
  L_{00}&L_{0\bm{b}} & 0 \\
  L_{\bm{a}0} & L_{\bm{a}\bm{b}} & 0 \\
  0 & 0 & O
 \end{array}
 \right)
 =\left(
 \begin{array}{cc}
  L_{\alpha\beta} & 0\\
  0  & O
 \end{array}
 \right),\quad
 \alpha=(0,\bm{a}), \, \beta=(0,\bm{b}),
 \label{HessianLsimple}
\end{align}
then
we get simpler formula for the Berwald functions.
We define a new matrix:
\begin{align}
 \bigl(\tilde{L}^{ab}\bigr)&:=\left(
  \begin{array}{ccc}
   \tilde{L}^{00} & \tilde{L}^{0\bm{b}} & 0 \\
   \tilde{L}^{\bm{a}0} & \tilde{L}^{\bm{a}\bm{b}} & 0 \\
   0 & 0 & O
  \end{array}
 \right)=\left(
 \begin{array}{cc}
  \tilde{L}^{\alpha\beta} & 0\\
  0  & O
 \end{array}
 \right),
 \\
 \bigl(\tilde{L}^{\alpha\beta}\bigr)&:=\left(
  \begin{array}{cc}
  \displaystyle 
   \left(\frac{\theta^0}{L}\right)^2 L^{\bm{a}\bm{b}}L_{\bm{a}}L_{\bm{b}} &
  \displaystyle 
   -\frac{\theta^0}{L}\left(L^{\bm{b}\bm{c}}L_{\bm{c}}
   -\frac{L^{\bm{a}\bm{c}}L_{\bm{a}}L_{\bm{c}}\theta^{\bm{b}}}{L}\right) \\
  \displaystyle  
   -\frac{\theta^0}{L}\left(L^{\bm{a}\bm{c}}L_{\bm{c}}
   -\frac{L^{\bm{b}\bm{c}}L_{\bm{b}}L_{\bm{c}}\theta^{\bm{a}}}{L}\right)  & 
  \displaystyle  
   L^{\bm{a}\bm{b}}-\frac{L^{\bm{a}\bm{c}}L_{\bm{c}}\theta^{\bm{b}}}{L}
   -\frac{L^{\bm{b}\bm{c}}L_{\bm{c}}\theta^{\bm{a}}}{L}
   +\frac{L^{\bm{c}\bm{d}}L_{\bm{c}}L_{\bm{d}}\theta^{\bm{a}} \theta^{\bm{b}}}{L^2}
  \end{array}
 \right),\label{Ltilde}
\end{align}
which satisfies the following identities
\begin{align}
 \tilde{L}^{\alpha\gamma}L_{\gamma\beta}
 =\delta^\alpha_\beta-\frac{\theta^\alpha L_\beta}{L}, \quad
 \tilde{L}^{\alpha\gamma}L_\gamma=0.
\end{align}

\begin{pr}\label{prop43}
 The Berwald functions and constraints for the Finsler metric of type 
(\ref{HessianLsimple}) are
 \begin{align}
  &G^\alpha=
  -\frac12\tilde{L}^{\alpha\gamma}L_\beta \theta^\beta
  \bigl([e_b,e_\gamma]\bigr)\theta^b,
  \quad
  G^I=\lambda^I,
  \label{simpleformula}
  \\
  &{\cal C}_I=L_\beta \theta^\beta \bigl([e_b,e_I]\bigr)\theta^b=0.
 \label{simpleconst}
 \end{align}
\end{pr}
\begin{proof}
In this case, we can take 
$\displaystyle {v_I}^a=\delta_I^a-\frac{L_I \theta^a}{L}=\delta_I^a$.
Since $\theta^\gamma M_\gamma=\theta^0 M_0+\theta^{\bm{c}}M_{\bm{c}}=0$
for $M_\gamma:=-\frac12 L_\beta \theta^\beta \bigl([e_b,e_\gamma]\bigr)\theta^b$
and
\begin{align}
 \tilde{L}^{\bm{a}\gamma}M_\gamma&=-\frac{\theta^0}{L}\left(L^{\bm{a}\bm{c}}L_{\bm{c}}
 -\frac{L^{\bm{b}\bm{c}}L_{\bm{b}}L_{\bm{c}}\theta^{\bm{a}}}{L}\right)M_0
 +L^{\bm{a}\bm{c}}M_{\bm{c}}-\frac{L^{\bm{a}\bm{b}}L_{\bm{b}}\theta^{\bm{c}}}{L}M_{\bm{c}}
 \nonumber \\
 &\hspace{12pt}
   -\frac{L^{\bm{c}\bm{b}} L_{\bm{b}} \theta^{\bm{a}} }{L} M_{\bm{c}}
   +\frac{L^{\bm{c}\bm{d}} L_{\bm{c}} L_{\bm{d}} \theta^{\bm{a}} \theta^{\bm{c}}}{L^2}M_{\bm{c}}
 \nonumber 
 \\
 &=\frac{1}{L} \left( L^{\bm{a}\bm{c}} L_{\bm{c}}
 -\frac{L^{\bm{b}\bm{c}} L_{\bm{b}} L_{\bm{c}} \theta^{\bm{a}} }{L}
 \right) \theta^{\bm{c}} M_{\bm{c}}
 +L^{\bm{a}\bm{c}}M_{\bm{c}}-\frac{L^{\bm{a}\bm{b}}L_{\bm{b}}\theta^{\bm{c}}}{L}M_{\bm{c}}
 \nonumber \\
 &\hspace{12pt}
   -\frac{L^{\bm{c}\bm{b}} L_{\bm{b}} \theta^{\bm{a}} }{L} M_{\bm{c}}
   +\frac{L^{\bm{c}\bm{d}} L_{\bm{c}} L_{\bm{d}} \theta^{\bm{a}} \theta^{\bm{c}}}{L^2}M_{\bm{c}}
 \nonumber 
 \\
 &=L^{\bm{a}\bm{c}}M_{\bm{c}}-\frac{L^{\bm{c}\bm{b}} L_{\bm{b}} \theta^{\bm{a}} }{L} M_{\bm{c}}
  =\left(\delta^{\bm{a}}_{\bm{b}}-\frac{L_{\bm{b}}\theta^{\bm{a}}}{L}\right)
   L^{\bm{b}\bm{c}}M_{\bm{c}},
\end{align}
$G^{\bm{a}}$ geven by (\ref{simpleformula}) is the same as 
the Proposition \ref{formula}.
$G^0=\tilde{L}^{0\gamma}M_\gamma$ is also the same as the Proposition, since
\begin{align}
 \tilde{L}^{0\gamma}M_\gamma
 &=\left(\frac{\theta^0}{L}\right)^2L^{\bm{a}\bm{b}}L_{\bm{a}}L_{\bm{b}}
 M_0-\frac{\theta^0}{L}\left(L^{\bm{c}\bm{a}}L_{\bm{a}}
 -\frac{L^{\bm{a}\bm{b}}L_{\bm{a}}L_{\bm{b}}\theta^{\bm{c}}}{L}\right)M_{\bm{c}}
 \nonumber \\
 &=-\frac{\theta^0}{L}L^{\bm{c}\bm{a}}L_{\bm{a}}M_{\bm{c}}
 =\left(\delta^{0}_{\bm{b}}-\frac{L_{\bm{b}}\theta^{0}}{L}\right)
   L^{\bm{b}\bm{c}}M_{\bm{c}}.
\end{align}
Using (\ref{simpleconst}), we get the above formula for $G^I$.
\end{proof}

\section{Euler-Lagrange equations in moving frame}

Euler-Lagrange equations are derived from the variational principle of the action,
\begin{align}
 {\cal A}[\bm{c}]:=\int_{\bm{c}}L\bigl(\theta^a\bigr)
 =\int_{s_0}^{s_1}L\left({e^a}_\mu(x(s)) \frac{dx^\mu(s)}{ds}\right)ds,
\end{align}
where $\bm{c}$ means an oriented curve on Finsler manifold $(M,L)$, and
$c(s):[s_0,s_1] \to \bm{c} \subset M$ is an arbitrary parametrization of $\bm{c}$.
The variation of ${\cal A}[\bm{c}]$ is defined by
\begin{align}
 \delta {\cal A}[\bm{c}]:=\int_{\bm{c}} \delta L
 =\int_{\bm{c}} \bib{L}{\theta^a}\delta \theta^a
 =\int_{\bm{c}} \bib{L}{\theta^a}
  \left(\bib{{e^a}_\mu}{x^\nu}\delta x^\nu dx^\mu
  +{e^a}_\mu d\delta x^\mu\right).
\end{align}
With a vecor field $\displaystyle X:=\delta x^\mu \bib{}{x^\mu}$,
the variation $\delta \theta^a$ can be described as
\begin{align}
 \delta \theta^a=\bib{{e^a}_\mu}{x^\nu}\delta x^\nu dx^\mu
  +{e^a}_\mu d\delta x^\mu
 &=\bib{{e^a}_\mu}{x^\nu}\delta x^\nu dx^\mu-\bib{{e^a}_\mu}{x^\nu}dx^\nu \delta x^\mu
 +d{e^a}_\mu \delta x^\mu+{e^a}_\mu d\delta x^\mu
 \nonumber \\
 &=\iota_X {\rm rot \,} \theta^a+ {\rm grad \,} \iota_X \theta^a 
 ={\cal L}_X \theta^a,
\end{align}
where ${\rm grad}$ and ${\rm rot}$ represent the exterior derivatives
on 1-forms and 2-forms, respectively.
For a vector field $X=X^a e_a$ with a 
moving frame $\bigl(e_a\bigr)$
and an arbitrary coefficients $X^a:={e^a}_\mu \delta x^\mu$, we have
\begin{align}
 \delta \theta^a={\cal L}_X \theta^a
 &={\rm grad \,} \iota_X \theta^a+\iota_X {\rm rot \, }\theta^a
 ={\rm grad \,}X^a-\iota_X \left\{
  \frac12 \theta^a\bigl([e_b,e_c]\bigr)\theta^b \wedge \theta^c
 \right\} \nonumber \\
 &=dX^a -  X^b \theta^a \bigl([e_b,e_c]\bigr) \theta^c,
\end{align}
and accordingly,
\begin{align}
 \delta L = \bib{L}{\theta^a} \delta \theta^a 
 =d \left( \bib{L}{\theta^a} X^a \right)  
 -X^a \left\{ d\left(\bib{L}{\theta^a}\right) 
 +\bib{L}{\theta^b} \theta^b\bigl([e_a,e_c]\bigr) \theta^c 
 \right\}.
\end{align}
The variational principle $\delta {\cal A}[\bm{c}]=0$
leads to the following Euler-Lagrange equations in moving frame formulation,
\begin{align}
 0=c^\ast \bigl\{ dL_a 
 +L_b \, \theta^b\bigl([e_a,e_c]\bigr) \theta^c
 \bigr\}.\label{ELeq}
\end{align}
The equations above determine the extremum curve $\bm{c}$ as a solution.
With the Cartan condition (\ref{Cartan2}) of the 
nonlinear Finsler connection,
the terms in the bracket of the right-hand side 
of the Euler-Lagrange equations (\ref{ELeq}) become
\begin{align}
 0&=dL_a
    +L_b
    \left( \bib{{N^b}_a}{\theta^c}-\bib{{N^b}_c}{\theta^a}\right)\theta^c.
 \label{ELeq2}
\end{align}
Because of the fact that ${N^b}_a$ are homogeneous functions of degree one
with respect to $\theta$:
\begin{align}
 \bib{{N^b}_a}{\theta^c}\theta^c={N^b}_a,
\end{align}
and the metricity condition of the Finsler nonlinear connection (\ref{metricity2}):
\begin{align}
 0=L_b {N^b}_a,\label{metricity3}
\end{align}
the second term of (\ref{ELeq2}) vanishes.
Moreover, differentiating the condition (\ref{metricity3}) with respect to $\theta$,
we obtain
\begin{align}
 0=L_{ab}{N^b}_c+L_b\bib{{N^b}_c}{\theta^a}.
\end{align}
With the aid of the above relation and the definition of the derivative
\begin{align}
 \displaystyle d\left(\bib{L}{\theta^a}\right)=\bbib{L}{\theta^b}{\theta^a}d\theta^b,
\end{align}
the equations (\ref{ELeq2}) finally become
\begin{align}
 0=L_{ab}\bigl(d\theta^b+{N^b}_c \theta^c\bigr).
 \label{LabDtheta}
\end{align}
Therefore the Euler-Lagrange equations (\ref{ELeq}) are rewritten into
the form of the auto-parallel equations,
\begin{align}
 \left\{  
 \begin{array}{l}
 \medskip
 \displaystyle
 c^\ast \left\{ d\theta^a + 2G^a-\lambda^0 \frac{\theta^a}{L} \right\}=0, 
 \quad  G^a=\frac12 \theta^c {N^a}_c, \quad (a=0,1,2,\dots,n) \\
 \displaystyle
 c^\ast {\cal C}_I =c^\ast \left\{
 L_{I{\bm a}}\lambda^{\bm a}
 +\frac{1}{2} L_d \, \theta^d \bigl([e_b,e_I]\bigr) \theta^b\right\}=0,
 \quad  (I=n-D+1, n-D+2,\dots,n)
 \end{array}  
 \right. \label{ELeqauto}
\end{align}
where $\lambda^0$ is an arbitrary function of $(x^\mu,\theta^a)$
with degree two homogeneity with respect to $\theta^a$.
The presence of this arbitrary function $\lambda^0$ in the auto-parallel equations
represents their parameterization invariance.
The equations (\ref{LabDtheta}) tell that 
$d\theta^a+2G^a$ should be spaned by eigenvectors
for zero eigenvalue of $\displaystyle \bigl(L_{ab}\bigr)$,
$\displaystyle \frac{\theta^a}{L}$ and ${v_I}^a$.
$G^a$, however, already include $\lambda^I {v_I}^a$ with arbitrary $\lambda^I$ in itself,
which leads to the first equation of (\ref{ELeqauto}).
The second line of (\ref{ELeqauto}) represents constraints.
These constraints can be classified into two categories:
the case that $\lambda^I$ is fixed so that it is consistent with
the derivative of ${\cal C}_I$ (second class constraints),
or the case that $\lambda^I$ remains arbitrary (first class constraints).

\section{Examples}

We show two
applications in moving frame formalism: 
i) Riemannian manifold, 
and ii) pseudo particle model~\cite{OITY}. 

\subsection{Riemannian manifold}

It is a good exercise to go over the method
in the case of a Riemannian manifold $(M,g)$
as a Finsler manifold $(M,L)$.
The Finsler metric $L$ for a
Riemannian metric $g=g_{\mu\nu}(x)dx^\mu \otimes dx^\nu$
in terms of a moving frame $\theta^a$ is given by
\begin{align}
 L=\sqrt{g_{\mu\nu}(x)dx^\mu dx^\nu}=\sqrt{\eta_{ab} \theta^a \theta^b}.
 \label{Riemann-metric}
\end{align}
We already know the solution: the nonlinear Finsler connection 
\begin{align}
 \nabla \theta^a = -\theta^c \otimes {N^a}_c,
\end{align}
in this case is linear,
${N^a}_c={\omega^a}_{bc} \theta^b$, and 
the metricity condition is just $\omega_{abc}+\omega_{bac}=0$.
Here we rederive the result through 
our formula (\ref{simpleformula}) in Proposition \ref{prop43}
for the Berwald functions $G^a$.
We have
\begin{align}
 &L_a = \bib{L}{\theta^a}=\frac{\eta_{ab}\theta^b}{L},
 \quad 
 L_{ab}=\bbib{L}{\theta^a}{\theta^b}
 =\frac{1}{L}\left( \eta_{ab} - \frac{\eta_{ac}\theta^c\ \eta_{bd}\theta^d}{L^2}\right),
 \nonumber \\ 
 &\tilde{L}^{ab}
 = L \eta^{ab}-\frac{\theta^a \theta^b}{L}
 \quad  \tilde{L}^{ac} L_{cb}=\delta^a_b
 -\frac{\theta^a L_b}{L}, \quad \tilde{L}^{ac}L_c=0,
\end{align}
therefore 
the Berwald functions become
\begin{align}
G^a &= -\frac12 \tilde{L}^{ac}L_f \,\theta^f \bigl( 
 [e_b,e_c]\bigr) \theta^b
 = -\frac12 \eta^{ae} \eta_{fd} \theta^f \bigl([e_b,e_e]\bigr) \theta^b \theta^d,
 \label{Riemannian-B}
\end{align}
without constraints.
From (\ref{Riemannian-B}) and (\ref{GtoN}),
we obtain a linear connection
\begin{align}
 {N^a}_c &= \bib{G^a}{\theta^c}-\frac12 \theta^a\bigl([e_b,e_c]\bigr) \theta^b\nonumber\\
 &=-\frac12 \bigl\{
   \eta^{ae} \eta_{fc} \theta^f \bigl([e_b,e_e]\bigr) \theta^b 
  +\eta^{ae} \eta_{fd} \theta^f \bigl([e_c,e_e]\bigr) \theta^d 
  +\theta^a \bigl([e_b,e_c]\bigr) \theta^b
  \bigr\}=:{\omega^a}_c,       
\end{align} 
which is indeed the spin connection ${\omega^a}_c={\omega^a}_{bc}\theta^b$
on the Riemannian manifold.

The Euler-Lagrange equation of a free particle  
on the Riemannian manifold $(M,L)$ described by the Finsler metric
(\ref{Riemann-metric}) is 
\begin{align}
 0 &=c^\ast \left\{
 d\left(\frac{\theta_a}{L}\right)+\frac{1}{L}\theta^b \bigl([e_a,e_c]\bigr)
 \theta_b \theta^c \right\}, \quad 
 L_a=\frac{\theta_a}{L}, \quad
 \theta_a:=\eta_{ab}\theta^b.
\end{align}
This is 
written in the form of the auto-parallel equation
\begin{align}
 0=c^\ast \bigl\{
 d\theta^a+{\omega^a}_c \theta^c-\lambda \theta^a
 \bigr\},
\end{align}
with an arbitrary function $\lambda$ of $x$ and $dx$ which is a Lagrange multiplier
representing reparameterization invariant.

\subsection{Pseudoclassical particle}

The next
example is a 
generalization of the pseudoclassical point particle
model of Casalbuoni, Brink and Schwarz~\cite{Casalbuoni,Brink-Schwarz,Freund}.
We previously considered the model on a curved 
two dimensional spacetime 
in terms of super Finsler geometry~\cite{OITYsuper}.
In this section, we will recast a super Finsler connection of the model
in moving frame.
The Finsler metric of the model which we call 
Casalbuoni-Brink-Schwarz metric (CBS metric) is given by
\begin{align}
 L=\sqrt{g_{\mu\nu}(x)\Pi^\mu \Pi^\nu}, \quad
 \Pi^\mu=dx^\mu+\langle \xi|\gamma^\mu(x)|d\xi\rangle,
\end{align}
where $g_{\mu\nu}(x)$ is a Lorentz metric, 
$(x^\mu)=(x^0,x^1)$ are Gra{\ss}mann even coordinates,
$|\xi\rangle=|_A\rangle \xi^A=|_1\rangle \xi^1+|_2\rangle \xi^2$
is a 2-dimensional Majorana spinor, and
$(\xi^A)=(\xi^1,\xi^2)$ are Gra{\ss}mann odd coordinates,
$\gamma^\mu(x)={e_a}^{\mu}(x)\gamma^a$,
$(\gamma^a)=(\gamma^0,\gamma^1)$ are Dirac matrices:
$\gamma^a \gamma^b+\gamma^b \gamma^a=2\eta^{ab}$,
${e_a}^\mu(x)$ are zweibeins: $\eta^{ab}{e_a}^\mu(x){e_b}^\nu(x)=g^{\mu\nu}(x)$, 
and $\langle \,,\, \rangle$ is a metric of the spinor space:
$\langle_A|{}_B\rangle:=\delta_{AB}$.
We consider the CBS metric $L$ as a super Finsler metric on 
a supermanifold $M^{(2,2)}$ with even and odd coordinates 
$(x^\mu,\xi^A)=:(z^I)$;
\begin{align}
 L=L(z,dz):\bm{v} \in T_pM^{(2,2)} \mapsto 
 L\bigl(z(p),dz(\bm{v})\bigr) \in \mathbb{R}_c, 
\end{align}
where $\mathbb{R}_c=\mathbb{R}_{S[2\,0]}$
is a set of even real numbers included in
a Gra{\ss}mann algebra over $\mathbb{R}$ with two odd generaters
$\Lambda_2=\mathbb{R}_{S[2]}$~\cite{DeWitt, Rogers, Freund}.

We define a super moving coframe and frame 
$\Theta^\Psi=\bigl(\Theta^a, \Theta^A \bigr), \, E_\Psi=\bigl(E_a,E_A\bigr)$:
\begin{align}
 &\Theta^a:={e^a}_\mu dx^\mu+\langle \xi|\gamma^a|d\xi\rangle, \quad
 E_a:={e_a}^\mu \bib{}{x^\mu}, \quad (a=1,2), \\
 &\Theta^A:=d\xi^A, \quad 
 E_A:=\bib{}{\xi^A}
 -\langle \xi|\gamma^a|_A \rangle  {e_a}^\mu  \bib{}{x^\mu},
 \quad (A=1,2),
\end{align}
then the CBS metric is written by
\begin{align}
 L=L(\Theta^a)=\sqrt{\eta_{ab}\Theta^a \Theta^b}.
\end{align}
The super Finsler connection $N^\Psi=\bigl({N^a}, {N^A}\bigr)$ is defined by
\begin{align}
 &\nabla \Theta^\Psi=-\Theta^\Omega \otimes {N^\Psi}_\Omega, \\
 &
 \bib{{N^\Psi}_\Omega}{\Theta^\Phi}
  -(-1)^{|\Omega||\Phi|}\bib{{N^\Psi}_\Phi}{\Theta^\Omega}
  +\Theta^\Psi \bigl([E_\Phi,E_\Omega]\bigr)=0,
 \\
 &\nabla L=\nabla \Theta^a \bib{L}{\Theta^a}
 =-\Theta^\Omega \otimes {N^a}_\Omega \bib{L}{\Theta^a}
 =0,
\end{align}
where index $\Psi, \Phi, \Omega$ takes $(a,A)$, and 
${N^\Psi}_\Omega$ are functions of 
$(x^\mu,\xi^A,\Theta^a,\Theta^A)$ and homogeneity one with respect to 
$(\Theta^a, \Theta^A)$,
that is 
$\displaystyle {N^\Psi}_{\Omega}=\bib{{N^\Psi}_{\Omega}}{\Theta^\Phi}\Theta^\Phi$.

Since $L_A=0$,
the Berwald functions and the constraints for this metric become
\begin{align}
 &G^a=-\frac12 \tilde{L}^{ac}L_b \, \Theta^b \bigl([E_\Phi,E_c]\bigr)\Theta^\Phi,\quad
 G^A=\lambda^A, \\
 &{\cal C}_A=L_b \, \Theta^b \bigl([E_A,E_\Phi]\bigr)\Theta^\Phi=0,
\end{align}
from the Proposition \ref{prop43}.
$\lambda^A$ are arbitrary functions of $\bigl(z,\Theta\bigr)$ 
and homogeneous two with respect to $\Theta^\Psi$, and
\begin{align}
 & L_a=\frac{\Theta_a}{L}, \quad L_A=0, \quad
 \bigl(L_{\Psi\Phi}\bigr)=
 \left(
  \begin{array}{cc}
  \displaystyle
  \frac{\eta_{ab}}{L}-\frac{\Theta_a \Theta_b}{L^3} & 0 \\
  0 & O
  \end{array}
 \right), \quad
 \tilde{L}^{ab}=L\eta^{ab}-\frac{\Theta^a \Theta^b}{L},
 \\
 &[E_A,E_b]=-[E_a,E_b]\langle\xi|\gamma^a|_A\rangle,
 \\
 &[E_A,E_B]=-2E_a \langle_B|\gamma^a|_A\rangle+[E_a,E_b]\langle \xi|\gamma^a|_A\rangle
  \langle \xi|\gamma^b|_B\rangle,
\end{align}
where $E_a$ are usual zweibeins on the 2-dimensional Lorentzian manifold.
Furthermore
\begin{align}
 \Theta^b \bigl([E_\Phi,E_c]\bigr)\Theta^\Phi
 &=\Theta^b \bigl([E_d,E_c]\bigr)\Theta^d
 +\Theta^b \bigl([E_D,E_c]\bigr)\Theta^D
 \nonumber 
 \\
 &=\theta^b \bigl([E_d,E_c]\bigr) \Theta^d
 -\theta^b\bigl([E_d,E_c]\bigr) \langle \xi|\gamma^d|_D\rangle \Theta^D
 \nonumber 
 \\
 &=\theta^b\bigl([E_d,E_c]\bigr)\theta^d
 +\theta^b\bigl([E_d,E_c]\bigr)\langle \xi|\gamma^d|d\xi\rangle
 -\theta^b\bigl([E_d,E_c]\bigr) \langle \xi|\gamma^d|d\xi\rangle
 \nonumber 
 \\
 &=\theta^b\bigl([E_d,E_c]\bigr)\theta^d,
 \\
 \Theta^b \bigl([E_A,E_\Phi]\bigr)\Theta^\Phi
 &=\Theta^b \bigl([E_A,E_c]\bigr)\Theta^c
 +\Theta^b \bigl([E_A,E_C]\bigr)\Theta^C
 \nonumber
 \\
 &=-\theta^b\bigl([E_d,E_c]\bigr)\langle \xi|\gamma^d|_A\rangle \Theta^c
 -2 \langle_A|\gamma^b|_C\rangle \Theta^C
 \nonumber \\
 & \hspace{24pt}
 +\theta^b\bigl([E_d,E_c]\bigr)\langle \xi|\gamma^d|_A\rangle
  \langle \xi|\gamma^c|_C\rangle \Theta^C,
 \nonumber 
 \\
 &=-\theta^b\bigl([E_d,E_c]\bigr)
 \langle \xi|\gamma^d|_A\rangle \theta^c
 -\theta^b\bigl([E_d,E_c]\bigr)\langle \xi|\gamma^d|_A\rangle
 \langle \xi|\gamma^c|d\xi \rangle
 \nonumber \\
 & \hspace{24pt}
 -2\langle _A |\gamma^b|d\xi\rangle
 +\theta^b\bigl([E_d,E_c]\bigr)\langle \xi|\gamma^d|_A\rangle
  \langle \xi|\gamma^c| d\xi \rangle,
 \nonumber
 \\
 &=
 -\theta^b\bigl([E_d,E_c]\bigr) \langle \xi|\gamma^d|_A\rangle \theta^c
 -2\langle _A |\gamma^b|d\xi\rangle, 
\end{align}
where we take $\theta^a:={e^a}_\mu dx^\mu$.
Therefore we get explicit forms of $G^a$ and ${\cal C}_A$, 
\begin{align}
 G^a&=-\frac12 \tilde{L}^{ac}L_b \, \theta^b\bigl([E_d,E_c]\bigr)\theta^d
 =-\frac12 \eta^{ac}\Theta_b \theta^b\bigl([E_d,E_c]\bigr)\theta^d
 +\frac12 \frac{\Theta^a \Theta^c}{L^2}\Theta_b \theta^b\bigl([E_d,E_c]\bigr)\theta^d
 \nonumber 
 \\
 &=
 -\frac12 \eta^{ac}\theta_b \theta^b\bigl([E_d,E_c]\bigr)\theta^d
 -\frac12 \eta^{ac}\theta^b\bigl([E_d,E_c]\bigr)\theta^d
  \langle \xi|\gamma_b|d\xi\rangle
 \nonumber 
 \\
 & \hspace{24pt}
 +\frac12 \frac{\Theta^a \Theta_b}{L^2} \theta^b\bigl([E_d,E_c]\bigr)\theta^d
 \langle \xi|\gamma^c|d\xi\rangle,
 \\
 {\cal C}_A&=\langle_A |\left\{
 -\frac{\Theta_b}{L}\theta^b\bigl([E_d,E_c]\bigr)\theta^c\gamma^d|\xi\rangle
 -\frac{2\Theta_b}{L}\gamma^b|d\xi\rangle
 \right\}=0.
\end{align}
If we have another vielbein basis: 
$dz^I=\bigl(dx^\mu, d\xi^A\bigr)$,
the natural coordinates basis,
their nonlinear connections coefficients
${n^I}_K$ with respect to this basis, are defined by
\begin{align}
 -dz^K \otimes {n^\mu}_K &=\nabla dx^\mu
 =\nabla \bigl\{ {e_a}^\mu \bigl(\Theta^a - \langle \xi|\gamma^a|\Theta\rangle \bigr)\bigr\} 
 \nonumber 
 \\
 &=d{e_a}^\mu \otimes \bigl(\Theta^a-\langle \xi|\gamma^a|\Theta\rangle\bigr)
 -\Theta^\Omega \otimes {e_a}^\mu {N^a}_\Omega
 -d\xi^A \otimes \langle_A|\gamma^\mu|\Theta\rangle
 \nonumber \\
 & \qquad
 +(-1)^{|\Omega|}\Theta^\Omega \otimes \langle \xi|\gamma^\mu|N_\Omega\rangle
 \\
 -dz^K \otimes {n^A}_K &= \nabla d\xi^A
 =\nabla \Theta^A
 =-\Theta^\Omega \otimes {N^A}_\Omega,
\end{align}
so that we have 
\begin{align}
 {n^\mu}_\rho &=-\bib{{e_a}^\mu}{x^\rho}\theta^a
 +{e^c}_\rho {e_a}^\mu {N^a}_c
 -{e^c}_\rho \langle \xi|\gamma^\mu|N_c\rangle,
 \\
 {n^\mu}_C&=-\langle \xi|\gamma^c|_C\rangle {e_a}^\mu {N^a}_c
 +{e_a}^\mu {N^a}_C+\langle_C|\gamma^\mu|\Theta\rangle
 +\langle \xi|\gamma^c|_C\rangle \langle \xi|\gamma^\mu|N_c\rangle
 +\langle \xi|\gamma^\mu|N_C\rangle,
 \\
 {n^A}_\rho&={e^c}_\rho {N^A}_c, \\
 {n^A}_C&=-\langle \xi|\gamma^c|_C\rangle {N^A}_c 
 +{N^A}_C.
\end{align}
Then we obtain
\begin{align}
 2g^\mu&:=dz^K {n^\mu}_K 
 =-d{e_a}^\mu \theta^a+\theta^c {e_a}^\mu {N^a}_c
 -\theta^c \langle \xi|\gamma^\mu|N_c\rangle
 +\langle \xi|\gamma^c|d\xi\rangle{e_a}^\mu {N^a}_c
 \nonumber 
 \\
 & \quad
 +{e_a}^\mu d\xi^C {N^a}_C+\langle d\xi|\gamma^\mu|\Theta\rangle
 -\langle \xi|\gamma^c|d\xi\rangle \langle \xi|\gamma^\mu|N_c\rangle
 -\langle \xi|\gamma^\mu|d\xi^C N_C\rangle
 \nonumber 
 \\
 &=-d{e_a}^\mu \theta^a+\langle \Theta|\gamma^\mu|\Theta\rangle+
 {e_a}^\mu \bigl(\Theta^c{N^a}_c+\Theta^C {N^a}_C
 \bigr)-\langle \xi|\gamma^\mu|\bigl(\Theta^c N_c+\Theta^CN_C \bigr)\rangle
 \\
 &=-d{e_a}^\mu \theta^a 
 +2{e_a}^\mu G^a-2\langle \xi|\gamma^\mu|G\rangle,
 \\
 2g^A&:=dz^K {n^A}_K=\theta^c{N^A}_c 
 +\langle \xi|\gamma^c| d\xi \rangle {N^A}_c
 +d\xi^C {N^A}_C
 \nonumber 
 \\
 &=\Theta^c {N^A}_c+\Theta^C {N^A}_C =2G^A.
\end{align}
The identities
\begin{align}
 &\langle d\xi|{\cal C}\rangle
 =-\frac{\Theta_b}{L} \theta^b\bigl([E_d,E_c]\bigr)\theta^c
 \langle d\xi|\gamma^d|\xi\rangle, 
 \\
 &{e_a}^\mu \eta^{ac}\theta^b \bigl([E_d,E_c]\bigr)\theta^d
 =g^{\mu\rho}{e^c}_\rho \theta^b \bigl([E_d,E_c]\bigr)\theta^d
 =g^{\mu\rho} \iota_{\bib{}{x^\rho}} \left\{-\frac12
 \theta^b \bigl([E_d,E_c]\bigr)\theta^d \wedge \theta^c
 \right\}, \notag \\
 & \hspace{36mm}
 = g^{\mu\rho} \iota_{\bib{}{x^\rho}} {\rm div} \theta^b
 \\
 &d{e_a}^\mu \theta^a+{e_a}^\mu \eta^{ac}\theta_b \theta^b\bigl([E_d,E_c]\bigr)\theta^d
 = - \Gamma^\mu {}_{\nu \rho} dx^\nu dx^\rho,
\end{align}
leads to
\begin{align}
 g^\mu&=-\frac12 d{e_a}^\mu \theta^a
 +{e_a}^\mu G^a-\langle \xi|\gamma^\mu|G\rangle
 \nonumber 
 \\
 &=-\frac12 d{e_a}^\mu \theta^a
 -\frac12 {e_a}^\mu \eta^{ac} \theta_b \theta^b\bigl([e_d,e_c]\bigr)\theta^d
  -\frac12 {e_a}^\mu \eta^{ac} \theta^b\bigl([e_d,e_c]\bigr)\theta^d \langle \xi|\gamma_b|d\xi\rangle
 \nonumber 
 \\
 & \quad
  +\frac{\Theta^a\Theta_b}{2L^2} {e_a}^\mu \theta^b\bigl([e_d,e_c]\bigr)\theta^d
   \langle \xi|\gamma^c|d\xi\rangle
  -\langle\xi|\gamma^\mu|\lambda\rangle,
 \\
 &=\frac12 {\Gamma^\mu}_{\nu\rho}dx^\nu dx^\rho
 +\frac{1}{2L}\Pi^\mu \langle d\xi|{\cal C}\rangle-\frac12 g^{\mu\rho}
 \iota_{\bib{}{x^\rho}}{\rm div}\theta^a \langle \xi|\gamma_a|d\xi\rangle
 -\langle \xi|\gamma^\mu|\lambda\rangle,
 \\
 g^A&=\lambda^A,
\end{align}
where $\Pi^\mu={e_a}^\mu \Theta^a$, 
${\rm div}\theta^a=d{e^a}_\mu \wedge dx^\mu$.
These are indeed the results of Berwald functions described in the
previous paper.
As well as the example A, the motion of the free particle on the Finsler manifold
is given by
\begin{align}
 0=c^\ast \left\{ d\Theta^a+2G^a-\lambda \Theta^a\right\}, 
 \quad 
 0=c^\ast{\cal C}_A.
\end{align}

\section{Discussions}

Calculations of Finsler connection and curvature in holonomic coordinates
are in general complicated and often hinder theoretical progress in physical
applications of Finsler geometry, especially when one tries to understand
a theory which is an extension of general relativity.
It is, however, conceivable that moving frame (vielbein) formalism
makes it easier as in general relativity.
In this paper, we gave the calculation method in moving frame and
displayed how it freely works in two different examples.
We believe the method even helps to simplify the description of
much more intricate theories that appear in physics.


\end{document}